\documentclass[letterpaper,10pt,conference]{ieeeconf}  

\IEEEoverridecommandlockouts                              
\usepackage{mathptmx} 
\usepackage{times} 
\usepackage{amsmath} 
\usepackage{amssymb}  

\usepackage[noadjust,sort]{cite}

\let\labelindent\relax \usepackage{enumitem}

\title{\LARGE \bf Cost-Optimal Switching Protection Strategy in Adaptive Networks}

\author{Masaki Ogura and Victor M.~Preciado
\thanks{The authors are with the Department of Electrical and Systems
Engineering, University of Pennsylvania, Philadelphia, PA 19014, USA.
Email:  {\tt\small \{ogura,preciado\}@seas.upenn.edu}}%
\thanks{{This work was supported in part by the NSF under grants CNS-1302222
and IIS-1447470.}}%
}

\newtheorem{definition}{Definition}[section]
\newtheorem{assumption}[definition]{Assumption}

\newtheorem{proposition}[definition]{Proposition}
\newtheorem{theorem}[definition]{Theorem}
\newtheorem{remark}[definition]{Remark}

\newtheorem{problem}[definition]{Problem}

\newcommand{\dN}{\,dN}

\usepackage{accents}
\newcommand{\ubar}[1]{\underaccent{\bar}{#1}}

\DeclareMathOperator*{\col}{col}
\renewcommand{\Pr}{P}

\DeclareSymbolFont{bbold}{U}{bbold}{m}{n}
\DeclareSymbolFontAlphabet{\mathbbold}{bbold}
\newcommand{\onev}{\mathbbold{1}}

\usepackage{calrsfs}
\DeclareMathAlphabet{\pazocal}{OMS}{zplm}{m}{n}
\renewcommand{\mathcal}[1]{\pazocal{#1}}

\usepackage[usenames]{xcolor}

\DeclareMathOperator*{\minimize}{minimize}
\DeclareMathOperator*{\subjectto}{subject\ to}

\usepackage{subcaption}

\usepackage{mathtools}

\newcommand{\afterequation}{\vskip 3pt}

\allowdisplaybreaks[4]

\begin{document}

\maketitle
\thispagestyle{empty}
\pagestyle{empty}

\begin{abstract}
In this paper, we study a model of network adaptation mechanism to control spreading processes over switching contact networks, called adaptive susceptible-infected-susceptible model. The edges in the network model are randomly removed or added depending on the risk of spread through them.  By analyzing the joint evolution of the spreading dynamics ``in the network" and the structural dynamics ``of the network", we derive conditions on the adaptation law to control the dynamics of the spread in the resulting switching network. In contrast with the results in the literature, we allow the initial topology of the network to be an arbitrary graph. Furthermore, assuming there is a cost associated to switching edges in the network, we propose an optimization framework to find the cost-optimal network adaptation law, i.e., the cost-optimal edge switching probabilities. Under certain conditions on the switching costs, we show that the optimal adaptation law can be found using convex optimization. We illustrate our results with numerical simulations.
\end{abstract}

\section{Introduction}

Accurate prediction and effective control of spreading dynamics over networks
are relevant problems in epidemiology and public health, computer malware, or
security of cyberphysical networks. Although we find many recent advances in the
field of network epidemiology \cite{Pastor-Satorras2015}, there are still many
open questions to transfer this knowledge to realistic epidemiological
situations. One fundamental result in the mathematical analysis of spreading in
networks is the close connection between the eigenstructure of the contact
network and epidemic
thresholds~\cite{Lajmanovich1976,Chakrabarti2008,VanMieghem2009a}. This result
enabled the authors in
\cite{Preciado2013,Preciado2013a,Preciado2013b,Preciado2014} to propose a convex
optimization framework to design the optimal distribution of  pharmaceutical
resources to control disease spread. This framework is specially adapted to
static network structures in which the pattern of interconnections does not
change over time. As we argue below, this may not be the case in many practical
situations.

Social distancing is one of the most important nonpharmaceutical approaches to
control disease spread over human contact networks~\cite{Bell2006,Funk2010}.
Examples of social distancing are, for instance, isolation of patients, school
closures, and avoidance of crowds. In spite of the obvious effect that such
behavior have on the dynamics of the spread, there is a lack of studies about
the role of social distancing in the spread of diseases over human contact 
networks. One of the reasons is that social distancing induces an adaptation of
the network structure that depends on the state of the infection. Although there
are results in the literature about disease spreading over time-varying networks
(see, e.g., \cite{Volz2009,Perra2012,Ogura2015a}), these works are based on the
assumption that the evolution of the network is independent of the state of the
individuals. In this paper, we propose a tractable framework to analyze the
co-evolution of the state-dependent network structure and the dynamics of the
spreading process taking place on it.

Most of the available studies of spreading processes over human networks with
social distancing have been relying on various unrealistic simplifying
assumptions. The authors in \cite{Gross2006,Zanette2008a,Guo2013,Tunc2014}
propose epidemic thresholds under the so-called mixing assumption; all the
individuals in a network interact randomly with each other. However, this
assumption is not satisfied in structured human populations. Although the
analysis in \cite{Valdez2012} does not rely on the mixing assumption, it relies
on the quantity called a reproduction number, whose validity for disease spread
over time-varying networks is not yet fully established~\cite{Holme2015a}.

This paper analyzes, without the mixing assumption, the dynamics of spreading
processes taking place in switching networks whose structure adapt to the state
of the spread. The disease spread is modeled by an extended version of the
well-known susceptible-infected-susceptible (SIS) model, which is called the
adaptive SIS model~\cite{Guo2013}. Without the mixing assumption employed
in~\cite{Guo2013}, we derive conditions under which the network adaptation is
able to protect against the spread of the disease. We furthermore use these
conditions to propose a cost-optimal adaptation policy to contain the disease.
This policy is based on the assumption that adapting the network structure to
the state of the disease has an associated cost. The optimal policy can be then
found by solving an optimization program. Under certain conditions, this
optimization program can be effectively solved using elements from convex
optimization~\cite{Boyd2007}.

This paper is organized as follows. In Section~\ref{sec:model}, we introduce the
adaptive SIS model studied in this paper. In Section~\ref{sec:stability},  we
analyze the exponential stability of the infection-free equilibrium of the
adaptive SIS models. Based on our stability analysis, Sections~\ref{sec:homo}
and \ref{sec:hetero} study an cost-optimal adaptation strategy for networks of
homogeneous and heterogeneous agents, respectively.

\subsection{Mathematical Preliminaries}

The probability of an event is denoted by $\Pr(\cdot)$. The expectation of a random variable is denoted by~$E[\cdot]$. We let $I$ denote the identity matrix and $\onev_p$ the $p$-dimensional vector whose entries are all one (we omit the dimension~$p$ when it is obvious from the context). A real matrix $A$, or a vector as its special case, is said to be nonnegative, denoted by $A\geq 0$, if all the entries of $A$ are nonnegative. The notations~$A > 0$, $A\leq 0$ and $A<0$ are understood in the obvious manner. For another matrix $B$ having the same dimensions as $A$, the notation $A\leq B$ implies $A-B\leq 0$. We again understand $A<B$, $A\geq B$, and $A>B$ in the obvious manner. The Kronecker product~\cite{Brewer1978} of $A$ and $B$ is denoted by $A\otimes B$. Let $A$ be a square matrix. The maximum real part of the eigenvalues of $A$ is denoted by~$\eta(A)$. We say that $A$ is Hurwitz stable if $\eta(A)< 0$. Also, we say that $A$ is Metzler if the off-diagonal entries of $A$ are all non-negative. We say that $A$ is irreducible if no similarity transformation by a permutation matrix makes $A$ into a block upper triangular matrix. For matrices $A_1$, $\dotsc$, $A_n$, the direct sum $\bigoplus_{i=1}^n A_i$ is defined as the block diagonal matrix having the block diagonals $A_1$, $\dotsc$, $A_n$. When $A_1$, $\dotsc$, $A_n$ have the same number of columns, we define
$\col_{1\leq i\leq n} A_i = \col(A_1, \dotsc, A_n) $ as the block matrix obtained by stacking the matrices $A_1$, $\dotsc$, $A_n$.

A directed graph is a pair~$\mathcal G = (\mathcal V, \mathcal E)$, where
$\mathcal V$ is a finite set of nodes, and $\mathcal E \subseteq \mathcal
V\times \mathcal V$ is a set of directed edges. Unless otherwise stated, we
assume $\mathcal V = \{1, \dotsc, n\}$. A directed path from $i$ to $j$ in
$\mathcal G$ is an ordered set of nodes $(i_0, \cdots, i_\ell)$ such that $i_0 =
i$, $(i_k, i_{k+1})\in\mathcal E$ for $k=0, \dotsc, \ell-1$, and $i_\ell = j$.
We say that $\mathcal G$ is strongly connected if there exists a directed path
from $i$ to $j$ for all $i, j\in \mathcal V$. The adjacency matrix of $\mathcal
G$ is defined as the $n\times n$ matrix $A = [a_{ij}]_{i, j}$ such that $a_{ij}
= 1$ if $(i, j)\in\mathcal E$ and $a_{ij} = 0$ otherwise. Similarly, an
undirected graph is a pair \mbox{$\mathcal G = (\mathcal V, \mathcal E)$}, where
$\mathcal V$ is a finite set and $\mathcal E$ is a subset of unordered pairs
$\{i, j\}$ of the elements $i, j\in \mathcal V$. The adjacency matrix of an
undirected graph is defined in a similar manner. A graph is strongly connected
if and only if its adjacency matrix is irreducible.

Finally, we recall basic facts about a class of optimization problems called geometric programs~\cite{Boyd2007}. Let $x_1$, $\dotsc$, $x_m$ denote $m$ real positive variables. We say that a real-valued function $f$ of $x = (x_1, \dotsc, x_m)$ is a {\it monomial function} if there exist $c>0$ and $a_1, \dotsc, a_m \in \mathbb{R}$ such that $f(x) = c {\mathstrut x}_1^{a_1} \dotsm {\mathstrut x}_m^{a_m}$. Also, we say that $f$ is a {\it posynomial function} if it is a sum of monomial functions of $x$. Given posynomial functions $f_0$, $\dotsc$, $f_p$ and monomial functions $g_1$, $\dotsc$, $g_q$, the optimization problem
\begin{equation*}
\begin{aligned}
\minimize_x\ 
&
f_0(x)
\\
\subjectto\ 
&
f_i(x)\leq 1,\quad i=1, \dotsc, p, 
\\
&
g_j(x) = 1,\quad j=1, \dotsc, q, 
\end{aligned}
\end{equation*}
is called a {\it geometric program}. It is known~\cite{Boyd2007} that a geometric program can be converted into a convex optimization problem.

\section{Susceptible-Infected-Susceptible Model\\over Adaptive Networks}\label{sec:model}

This section introduces the model of spreading processes over
adaptive networks studied in this paper and state the optimal design problem
under consideration. Each node in the network can be in one of two states:
\emph{susceptible} or \emph{infected}. The state of node $i$ evolves over time
and is represented by a binary variable $x_i\in\{0, 1\}$. We say that node $i$
is {susceptible} at time $t$ if $x_i(t) = 0$, and is {infected} at time $t$ if
$x_i(t) = 1$. In this paper, we model the evolution of $x_i$ as a
continuous-time stochastic process taking values in $\{0, 1\}$. We also assume
that the structure of the network in which the spreading process is taking place
evolves over time. In particular, we model the network $\mathcal G$ as a
continuous-time random graph process taking values in the set of undirected
graphs with $n$ nodes. In other words, we model the dynamics of spreading as a
stochastic process taking place over a random graph process. We denote by
$\mathcal N_i(t)$ the set of neighbors of node $i$ in the graph $\mathcal G(t)$,
i.e., $\mathcal N_i(t)=\{j\in\mathcal{V}\colon \{i,j\}\in\mathcal G(t)\}$, and
by $A(t) = [a_{ij}(t)]_{i,j}$ the adjacency matrix of $\mathcal G(t)$.

The spreading models over adaptive networks studied in this paper are formally introduced as the class of pairs $(x, \mathcal G)  = (\{x_i\}_{i=1}^n, \mathcal G)$ satisfying the following definition:
 
\begin{definition}
Let $\mathcal G_0 = (\mathcal V, \mathcal E_0)$ be an undirected graph with
adjacency matrix $A_0 = [a_{ij}(0)]_{i,j}$. The pair $(x,\mathcal G)$ is said to
be an \emph{adaptive susceptible-infected-susceptible model} over $\mathcal G_0$
(\emph{ASIS model} for short) if there exist nonnegative numbers $\beta_i$,
$\delta_i$, $\phi_i$, and $\psi_{ij}$ ($i, j=1, \dotsc, n$) such that the
following conditions hold:
\begin{enumerate}[leftmargin=.5cm,label=\alph*)]
\item $\mathcal G(0) = \mathcal G_0$; 

\item The process $(x, \mathcal G)$ is Markov; 

\item For every $i$, the transition probabilities of $x_i$ are given
by \newcommand{\mi}{-.66cm}
\begin{align}
\hspace{\mi}\Pr(x_i(t+h) = 1 \mid x_i(t) = 0) &=
\beta_i\ \sum_{\mathclap{k \in \mathcal N_i(t)}}\ x_k(t)\,h + o(h),
\label{eq:infection}
\\
\hspace{\mi}\Pr(x_i(t+h) = 0 \mid x_i(t) = 1) &= \delta_i \,h + o(h), 
\label{eq:recovery}
\end{align}
where $o(h)$ is a function such that $\lim_{h\to 0}o(h)/h = 0$.

\item For all $i, j$, the transition probabilities of $a_{ij}$ are given by
\begin{align}
\hspace{\mi}\Pr(a_{ij}(t+h) \!=\! 0 \!\mid\! a_{ij}(t) \!=\! 1)
&= 
(\phi_ix_i(t) + \phi_{\!j}x_{\!j}(t)) h +  o(h),
\label{eq:cut}
\\
\hspace{\mi}\Pr(a_{ij}(t+h) \!=\! 1 \!\mid\! a_{ij}(t) \!=\! 0) &= a_{ij}(0)\psi_{ij} h +  o(h). 
\label{eq:rewire}
\end{align}

\item $\psi_{ij} = \psi_{ji}$ for all $i$ and $j$.
\end{enumerate}
The constants $\beta_i$, $\delta_i$, $\phi_i$, and $\psi_{ij}$ are respectively
called \emph{infection}, \emph{recovery}, \emph{cutting}, and \emph{rewiring
rates}.
\end{definition}

We can interpret the above model as follows. Item b) indicates that the future
evolution of the spread, given the present state, does not depend on the past.
The probabilities in c) describe how nodal states evolve. Notice that, if
$\mathcal G(t)$ were a static network, these probabilities would coincide with
those of the NIMFA model \cite{VanMieghem2009a} with heterogeneous infection and
recovery rates. Eqn.~\eqref{eq:infection} indicates that, if node~$i$ is
susceptible and its neighbor $j$ is infected, then $i$ becomes infected with the
instantaneous infection rate $\beta_i$. Moreover, the rate is proportional to
the number of infected neighbors. Eqn.~\eqref{eq:recovery} implies that, once
node $i$ becomes infected, it will become susceptible with an instantaneous
recovery rate $\delta_i$. 

Item d) describes an adaptation mechanism of the network to the state of the
disease. Eqn.~\eqref{eq:cut} indicates that, whenever a node $i$ is infected,
the node adaptively removes edges connecting the node and its neighbors
according to a Poisson process with rate $\phi_i$. This mechanism is designed to
contain the spread through edges connected to infected nodes. Moreover,
\eqref{eq:rewire} describes a mechanism for which `cut' edges are `rewired' or
added back to the network. We assume that edge~$\{i, j\}$ is added to the
network with a rewiring rate~$\psi_{ij}$. See Fig.~\ref{fig:adaptive} for a
schematic picture of these transition probabilities.
\begin{figure}[tb]
\begin{minipage}[b]{.5\linewidth}
\centering\includegraphics[width=.8\linewidth]{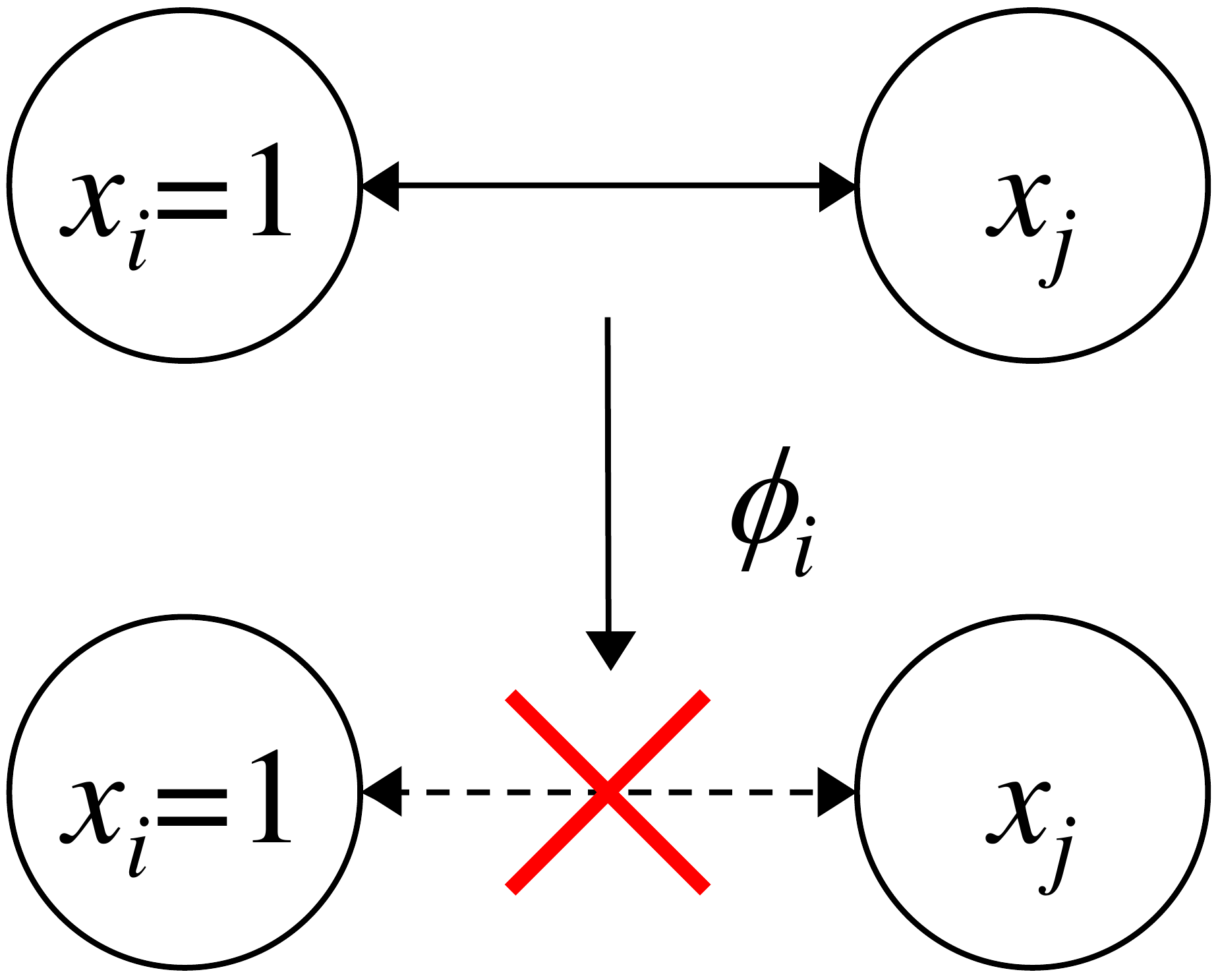}
\end{minipage}%
\begin{minipage}[b]{.5\linewidth}
\centering\includegraphics[width=.8\linewidth]{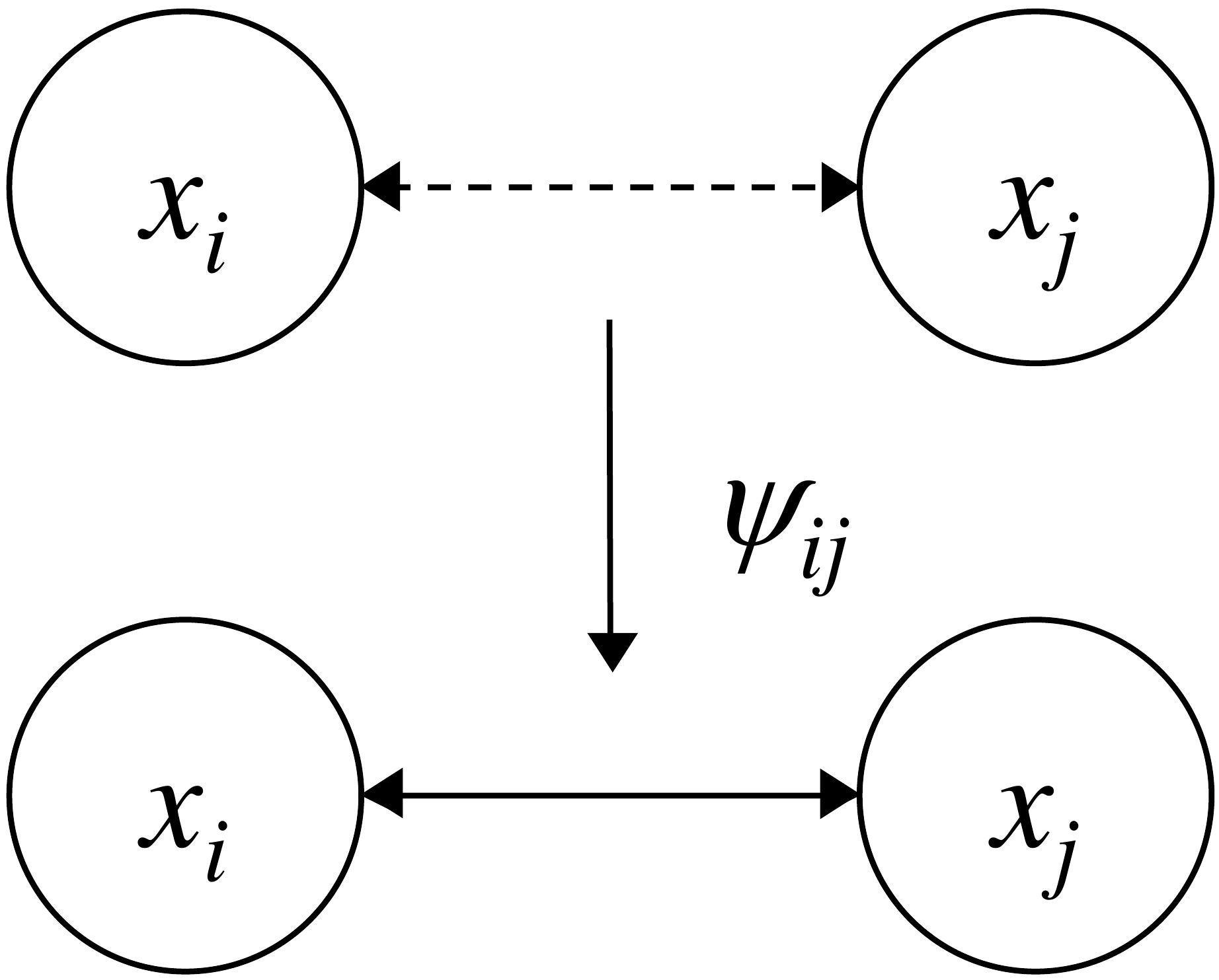}
\end{minipage}
\caption{Adaptively switching network.}
\label{fig:adaptive}
\end{figure}
Finally, Item e) follows from the assumption that $\mathcal G(t)$ is undirected,
although this is not an essential restriction and could be relaxed to account
for directed contact networks. Also, notice that we have included the term
$a_{ij}(0)$ in \eqref{eq:rewire} to guarantee that only those edges that were
present at the initial time $t = 0$ can be added later on by the rewiring
process.

\begin{remark}\label{rmk:}
{A model similar to the ASIS model proposed in this paper was studied in
\cite{Guo2013}, where it was assumed that
the initial graph~$\mathcal G_0$ is the complete graph. A major difference
between our model and the one in \cite{Guo2013} is the information available to each node. In the model in \cite{Guo2013}, it is assumed that nodes know the states of their neighbors. In contrast, we do not assume to have access to this knowledge in our model. This difference has a direct implication in the link-breaking process. For example, in \cite{Guo2013}, an
infected node does not break the edge between itself and its
infected neighbors. On the other hand, in our model, an infected node
will break edges independent of the state of its neighbors.}
\end{remark}

Once the adaptive network under consideration is described, we define the
exponential stability of the infection-free equilibrium $p_i(t) = 0$ of ASIS
models, as follows:

\begin{definition}
For $t\geq 0$ let $p_i(t) = \Pr(x_i(t) = 1)$ be the infection probability of
node~$i$. We say that the infection-free equilibrium $p_i(t) \equiv 0$ of the
adaptive SIS model $(x, \mathcal G)$ is \emph{exponentially stable}  if there
exist $K\geq 0$ and $\alpha > 0$ such that $p_i(t) \leq Ke^{-\alpha t}$ for all
$i$, $t$, and $x_i(0)$. We call $\alpha$ the \emph{decay rate}.
\end{definition}

In many practical situations, there is a cost associated to the mechanisms of cutting and rewiring edges in a network. Accordingly, we assume we have two scalar cost functions $f$ and $g$, defined on $[0, \infty)$, describing the cost associated to the rates of cutting and rewiring edges, respectively. The main purpose of this paper is to find the cost-optimal switching strategy, defined by the values of the cutting and rewiring rates, to drive the state of the spread towards the disease-free equilibrium at a given exponential rate. The total cost of a switching strategy is given by:
\begin{equation*}
C 
= 
\sum_{i=1}^nf(\phi_i) 
+  
\smashoperator[r]{\sum_{\{i, j\} \in \mathcal{E}_0}}\  g(\psi_{ij}). 
\end{equation*}
We also assume the following bounds on the rates: 
\begin{equation}\label{eq:bounds}
\ubar{\phi}\leq \phi_i \leq \bar{\phi}, \ 
\ubar{\psi}\leq \psi_{ij} \leq \bar{\psi}
\end{equation}
for some nonnegative numbers $\ubar{\phi}$, $\bar\phi$, $\ubar{\psi}$, and
$\bar{\psi}$. Now, we are ready to state the problem investigated in this
paper.\footnote{Since the design of infection and recovery rates have been
previously studied
in~\cite{Preciado2013,Preciado2013a,Preciado2013b,Preciado2014}, we focus our
attention on the design of $\phi_i$ and $\psi_{ij}$ only (although our framework
can be easily extended to include $\beta_i$ and~$\delta_i$ as additional design
variables). }

\begin{problem}\label{prb:}
Given $\alpha > 0$, find the cutting and rewiring rates $\phi_i$ and $\psi_{ij}$ satisfying \eqref{eq:bounds} such that the adaptive SIS model is exponentially stable with decay rate $\alpha$ and the total cost $C$ is minimized.
\end{problem}

In this paper we solve Problem~\ref{prb:} under the following reasonable assumption:

\begin{assumption}\label{assm:connected} 
$\mathcal G_0$ is strongly connected. Moreover, $\beta_i>0$, $\delta_i>0$, and $\psi_{ij} > 0$ for all $\{i, j\} \in \mathcal E_0$.
\end{assumption}

\section{Stability Analysis}\label{sec:stability}

In this section, we perform a stability analysis of the ASIS model $(x, \mathcal G)$ over $\mathcal G_0$. We begin by representing the model as a set of stochastic differential equations with Poisson counters. For $\gamma\geq 0$, we let $N_\gamma$ denote a Poisson counter with rate $\gamma$. We assume that all Poisson counters appearing in this paper are stochastically independent. We will use superscripts for the Poisson counters to distinguish those that has the same rates but are independent. Then, from \eqref{eq:infection} and \eqref{eq:recovery}, the evolution of the nodal states can be described as:
\begin{equation}\label{eq:dx_i}
dx_i = -x_i \dN_{\delta_i} + (1-x_i) \ \sum_{\mathclap{k\in\mathcal N_i(0)}}\ a_{ik}x_k \dN_{\beta_i}^{(k)}.
\end{equation}
Similarly, from \eqref{eq:cut} and \eqref{eq:rewire}, the evolution of the edges can be written as:
\begin{equation}\label{eq:da_ij}
da_{ij} = -a_{ij} (x_i \dN_{\phi_i}^{(j)} + x_{\!j} \dN_{\phi_{\!j}}^{(i)} )+ (1-a_{ij}) \dN_{\psi_{ij}},
\end{equation}
for all $i$ and $j$ such that $\{i, j\} \in \mathcal E_0$.

Using the stochastic differential equations~\eqref{eq:dx_i} and
\eqref{eq:da_ij}, we derive an upper bounding linear model for the infection
probabilities $p_i$. To state the linear model, we define the following
variables. Let us define $p(t) \in \mathbb{R}^n$ by $p = \col_{1\leq i\leq
n}p_i$. Also, for $i=1, \dotsc, n$ and $j\in \mathcal N_i(0)$, define $q_{ij}(t)
= E[a_{ij}(t)x_i(t)]$ and let $q_i = \col_{j\in \mathcal N_i(0)}q_{ij}$ and
{$q = \col_{1\leq i\leq n} q_i$}. Let $d_i$ denote the degree of node $i$
in the initial graph $\mathcal G_0$ and $m$ the number of the edges in $\mathcal
G_0$. Then, $q$ has the dimension $\sum_{i=1}^N d_i = 2m$. We also introduce the
following matrices. Define {$T_i\in\mathbb{R}^{1 \times (2m)}$} as the
unique matrix satisfying:
\begin{equation}\label{eq:def:T_i}
T_i q = \sum_{k \in \mathcal N_i(0)} q_{ki}. 
\end{equation}
Then define the matrices
$B_1 = \col_{1\leq i\leq n} (\beta_i T_i)$, 
$B_2 = \col_{1\leq i\leq n} (\beta_i \onev_{d_i} \otimes  T_i)$, 
$D_1 = \bigoplus_{i=1}^n\delta_i$, 
$D_2 = \bigoplus_{i=1}^n( \delta_i I_{d_i})$, 
$\Phi = \bigoplus_{i=1}^n (\phi_i I_{d_i})$, 
$\Psi_1 = \bigoplus_{i=1}^n \col_{j\in{\mathcal N}_i(0)}\psi_{ij}$, 
$\Psi_2 = \bigoplus_{i=1}^n 
\bigoplus_{{j\in\mathcal N_i(0)}}\;\psi_{ij}$.
Now, we can state the following theorem:

\begin{theorem}\label{thm:UpperBounds}
Define $M  \in \mathbb{R}^{(n+2m)\times(n+2m)}$ by
\begin{equation}\label{eq:M}
M  = 
\begin{bmatrix}
-D_1	&B_1
\\
\Psi_1	&B_2 - D_2 - \Phi - \Psi_2
\end{bmatrix}. 
\end{equation}
Then, for all $x_1(0)$, $\dotsc$, $x_n(0)$, it holds that
\begin{equation}\label{eq:d[pq]/dt}
\frac{d}{dt}\begin{bmatrix}
p \\ q
\end{bmatrix} 
\leq 
M 
\begin{bmatrix}
p \\ q
\end{bmatrix}. 
\end{equation}
\afterequation
\end{theorem}

\begin{proof}
Taking the expectations in \eqref{eq:dx_i} yields that
\begin{equation*}
\frac{d}{dt}E[x_i] 
= 
-\delta_i E[x_i] + \beta_i\ \ \sum_{\mathclap{k\in\mathcal N_i(0)}}\ \  E[(1-x_i)a_{ik}x_k]. 
\end{equation*}
Since $E[(1-x_i)a_{ik}x_k] \leq E[a_{ik}x_k] = q_{ki}$, from the definition
of~$T_i$ in~\eqref{eq:def:T_i}, we obtain ${dp_i}/{dt} \leq -\delta_i p_i + 
\beta_i T_i q$. This implies that $dp\!/dt \leq -D_1 p + B_1 q$, which proves
the upper half block of the inequality~\eqref{eq:d[pq]/dt}.

Then, let us evaluate $dq/dt$. The It\^o rule for jump
processes~\cite{Brockett2008} yields that
\begin{equation*}
\begin{aligned}
d(a_{ij}x_i)
\begin{multlined}[t]
=-a_{ij} x_i \dN_{\phi_i}^{(j)} - a_{ij}x_ix_{\!j}\dN_{\phi_{\!j}}^{(i)} + (1-a_{ij})x_i \dN_{\psi_{ij}}
\\
\hspace{.5cm}- a_{ij}x_i \dN_{\delta_i} + a_{ij}(1-x_i)\ \sum_{\mathclap{k\in\mathcal N_i(0)}}\  a_{ik}x_k \dN_{\beta_i}^{(k)}. 
\end{multlined}
\end{aligned}
\end{equation*}
Taking expectations in this equation, we obtain
\begin{equation}\label{eq:dqij/dt=...}
\begin{multlined}
\frac{dq_{ij}}{dt}
=
-\phi_i E[a_{ij}x_i] - \phi_{\!j} E[a_{ij}x_ix_{\!j}] + 
\\
\psi_{ij} E[(1-a_{ij})x_i] -\delta_i q_{ij} + \beta_i\ \sum_{\mathclap{k\in\mathcal N_i(0)}}\, E[a_{ij}(1-x_i)a_{ik}x_k].
\end{multlined}
\end{equation}
Since $\sum_{k\in\mathcal N_i(0)} E[a_{ij}(1-x_i)a_{ik}x_k] \leq
\sum_{{k\in\mathcal N_i(0)}} E[a_{ik}x_k] = T_i q$, we obtain ${dq_{ij}}/{dt}
\leq \psi_{ij} p_i - (\phi_i +\psi_{ij}+\delta_i)q_{ij} + \beta_i
\onev^\top_{d_i} T_i q$ from \eqref{eq:dqij/dt=...}. Stacking the variables
$q_{ij}$ for all $j\in \mathcal N_i(0)$ yields ${dq_i}/{dt} \leq
\col_{{j\in\mathcal N_i(0)}}(\psi_{ij}p_i) - (\phi_i + \delta_i) q_i - \psi_j
q_i + \beta_i( \onev_{d_i}\otimes T_i ) q$, where $\psi_j =
\bigoplus_{j\in\mathcal N_i(0)} \psi_{ji}$. This proves the lower half block of
the inequality~\eqref{eq:d[pq]/dt} and completes the proof.
\end{proof}

From Theorem~\ref{thm:UpperBounds} we immediately have the following sufficient condition for  exponential stability of the infection-free equilibrium.

\begin{theorem}\label{thm:stability}
If $M$ is Hurwitz stable, then the infection-free equilibrium of the adaptive SIS model is exponentially stable with a decay rate $-\eta(M)$.
\end{theorem}

Before closing this section, we prove the following proposition that plays an important role in the rest of the paper.

\begin{proposition}\label{prop:irreducible}
The matrix $M$ is irreducible. 
\end{proposition}

\begin{proof}
Define 
\begin{equation*}
L = \begin{bmatrix}
O & T \\ J&S
\end{bmatrix}, 
\end{equation*}
where
\begin{equation}\label{eq:JTS}
J = \bigoplus_{i=1}^n\onev_{d_i},\ 
T = \col_{1\leq i\leq n}T_i,\ 
S = \col_{1\leq i\leq n}(\onev_{d_i}\otimes T_i). 
\end{equation}
Since $\beta_i$ and $\psi_{ij}$ are positive by Assumption~\ref{assm:connected}, if $M_{ij} = 0$, then $L_{ij} = 0$ for all distinct $i$ and $j$. From this we see that, to show the irreducibility of $M$, it is sufficient to show that $L$ is irreducible.

In order to show that $L$ is irreducible, we shall show that the directed graph $\mathcal H$, defined as the graph having adjacency matrix~$L$, is strongly connected. We identify the nodes $1, \dotsc, n+2m$ of $\mathcal H$ using the variables $p_1$, $\dotsc$, $p_n$, $q_{1j}$ ($j\in\mathcal{N}_1(0)$), $\dotsc$, $q_{nj}$ ({$j\in\mathcal{N}_n(0)$}). Then, the upper-right block~$T$ of the matrix $L$ shows that the graph $\mathcal H$ has directed edge~$(p_i, q_{ji})$ for all $i=1, \dotsc, n$ and $j \in \mathcal N_i(0)$. Similarly, from the matrices $J$ and $S$, we see that $\mathcal H$ has the edges~$(q_{ij}, p_i)$ and $(q_{ij}, q_{ki})$ for all $i=1, \dotsc, n$ and {$j, k \in \mathcal N_i(0)$}. Then, let us show that $\mathcal H$ has a directed path from $p_i$ to $p_j$ for all $i, j\in \{1, \dotsc, n\}$. Since $\mathcal G_0$ is strongly connected, it has a path $(i_0, \dotsc, i_\ell)$ such that $i_0 = i$ and $i_\ell = j$. Therefore, from the above fact, we can see that $\mathcal H$ contains the directed path $(p_i, q_{i_1, i_0}, q_{i_2, i_1}, \dotsc, q_{i_\ell, i_{\ell-1}}, p_j)$. In the same way, we can show that $\mathcal H$ also contains the directed path $(p_i, q_{ji}, q_{ij}, p_i)$ for every $\{i, j\} \in \mathcal E_0$. These two observations show that $\mathcal H$ is strongly connected and, hence, $L$ is irreducible.
\end{proof}

\section{Homogeneous Case}\label{sec:homo}

Based on the stability analysis presented in the previous section, we study the optimal design problem stated in Problem~\ref{prb:}. We start our analysis by assuming that the ASIS model is homogeneous, as defined below (this restriction is relaxed in the next section):

\begin{definition}
We say that the adaptive SIS model is \emph{homogeneous} if there exist nonnegative constants $\beta$, $\delta$, $\phi$, and $\psi$ such that $\beta_i = \beta$, $\delta_i = \delta$, $\phi_i = \phi$, and $\psi_{ij} = \psi$ for all $i$ and $j$.
\end{definition}

In the homogeneous case, the stability criterion in Theorem~\ref{thm:stability} reduces to the next simple condition.

\begin{theorem}\label{thm:stbl:homo}
Assume that the adaptive SIS model is homogeneous. Let $\rho = \eta(A_0)$. Then, the infection-free equilibrium of the adaptive SIS model is exponentially stable if
\begin{equation}\label{eq:delta>...}
\delta 
> 
\frac{\beta\rho - \phi-\psi}{2} + \frac{\sqrt{(\beta\rho+\phi+\psi)^2 - 4\beta\rho\phi}}{2}. 
\end{equation}
\afterequation
\end{theorem}

\begin{proof}
Assume that the model is homogeneous. Then, the matrix $M$ defined in \eqref{eq:M} takes the form
\begin{equation*}
M = 
\begin{bmatrix}
-\delta I	& \beta T
\\
\psi J		& \beta S - (\delta + \phi + \psi) I
\end{bmatrix}, 
\end{equation*}
where the matrices $J$, $T$, and $S$ are defined by \eqref{eq:JTS}. We prove the
theorem under the assumption that $\beta\rho \neq \phi$. Since $\mathcal G_0$ is
strongly connected by Assumption~\ref{assm:connected}, $A_0$ is irreducible and
therefore has a positive eigenvector $v$ corresponding to the eigenvalue $\rho$
(see \cite{Farina2000}). Define the positive vector $w = \col_{1\leq i\leq
n}(v_i \onev_{d_i})$. Then, the definition of~$T_i$ shows $T_iw =
\sum_{k\in\mathcal N_i(0)} w_{ki} = \sum_{k\in\mathcal N_i(0)} v_k = (Av)_i =
\rho v_i$ and therefore $Tw = \lambda v$. In the same manner, we can show
{$Sw = \rho w$}. Since we have $Jv = w$,  for a nonnegative number $c$ it
follows that
\begin{equation}\label{eq:EigEq}
M
\begin{bmatrix}
cv\\w
\end{bmatrix}
=
\begin{bmatrix}
(\beta\rho - c\delta)v
\\
\left(c\psi + \beta \rho - (\delta + \phi + \psi)\right) w
\end{bmatrix}. 
\end{equation}
Hence, if a real number $\lambda$ satisfies the following equations:
\begin{equation}\label{eq:uations}
\beta\rho - c\delta = c\lambda,\ 
c\psi + \beta \rho - (\delta + \phi + \psi) = \lambda, 
\end{equation}
then, by \eqref{eq:EigEq}, we see that the nonnegative vector $\col(cv,w)$ is an eigenvector of the irreducible and Metzler matrix~$M$ corresponding to the eigenvalue $\lambda$. This implies that {$\eta(M) = \lambda$} (see \cite[Theorem~17]{Farina2000}). Therefore, the condition~{$\lambda <0$} is sufficient for exponential stability of the adaptive SIS model by Theorem~\ref{thm:stability}.

To find such $\lambda$, we solve \eqref{eq:uations} with respect to $\lambda$ and obtain $\lambda^2 + (2\delta + \phi + \psi - \beta\rho)\lambda + \delta(\delta+\phi+\psi)-\beta\rho(\delta + \psi)= 0$. This equation is satisfied by $\lambda = \lambda_+$, where
\begin{equation*}
\lambda_+ = \frac{\beta\rho-2\delta - \phi -\psi + \sqrt{(\beta\rho+\phi+\psi)^2 - 4\beta\rho\phi}}{2}. 
\end{equation*}
Then, the pair $(c, \lambda) = (\beta\rho/(\lambda_++\delta), \lambda_+)$
satisfies \eqref{eq:uations}. We remark that $\lambda_+ +\delta$ is positive
because of the initial assumption $\beta\rho \neq \phi$. Therefore, $c\geq 0$
and hence the above argument shows that
\begin{equation}\label{eq:eta=lam_+}
\eta(M) = \lambda_+. 
\end{equation}
Therefore, by Theorem~\ref{thm:stability}, the infection-free equilibrium of the adaptive SIS model is exponentially stable if $\lambda_+ < 0$, which is equivalent to \eqref{eq:delta>...}. 
\end{proof}

\begin{remark}
In the special case when the network does not adapt to the prevalence of infection, i.e., when $\phi=0$, Proposition~\ref{thm:stbl:homo} recovers the well-known stability condition $\delta > \beta \rho(A_0)$ for the SIS models over static networks \cite{Lajmanovich1976,VanMieghem2009a}. 
\end{remark}

The following theorem provides a solution to Problem~\ref{prb:}, in the homogeneous case:

\begin{theorem}\label{thm:optim:homo}
Assume that the adaptive SIS model is homogeneous. Let $\phi$ and $\psi$ be the solutions of the optimization problem:
\begin{align}
\minimize_{\phi,\psi}\ & nf(\phi) + mg(\psi)\notag
\\
\subjectto\ & \phi \geq (\beta\eta - \delta + 1)\left({\psi}/{(\delta-\alpha)} + 1\right), 
\label{eq:subjectto:homo}
\\
&\ubar{\phi} \leq \phi \leq \bar{\phi}, \ 
\ubar{\psi} \leq \psi \leq \bar{\psi}. \notag
\end{align}
Then, the pair $(\phi, \psi)$ gives the solution of Problem~\ref{prb:}.
\end{theorem}

\begin{proof}
It is sufficient to show that $\eta(M) \leq -\alpha$ if and only if the \eqref{eq:subjectto:homo} holds, but this easily follows from \eqref{eq:eta=lam_+}.
\end{proof}
%

\section{Heterogeneous Case}\label{sec:hetero}
In this section, we extend our analysis to non-homogeneous adaptive SIS models. We will show that Problem~\ref{prb:} can be effectively solved under the following assumption.

\begin{assumption}\ \label{assm:hetero}
\begin{enumerate}
\item The values of $\psi_{ij}$ are given for every  $\{i, j\} \in \mathcal E_0$;

\item There exist constants $r > \bar{\phi}$ and $s$ such that the function~$F \colon [r - \bar{\phi}, r - \ubar{\phi}] \to \mathbb{R}\colon x \mapsto s + f(r - x)$ is a posynomial function.
\end{enumerate}
\afterequation
\end{assumption}

In order to state the main result of this section, we will need the next proposition.

\begin{proposition}\label{prop:tildeA}
Let $\bar \delta = \max_i \delta_i$ and define {$\tilde \delta_i = \bar \delta - \delta_i$}. Similarly, let $\bar{\psi} = \max_{i, j}\psi_{ij}$ and define {$\tilde{\psi}_{ij} = \bar{\psi} - \psi_{ij}$}.  Let $\tilde \phi_1$, $\dotsc$, $\tilde \phi_n$ be real numbers. Define the matrices {$\tilde D_1 = \bigoplus_{i=1}^n \tilde \delta_i$}, $\tilde D_2 = \bigoplus_{i=1}^n (\tilde \delta_i I_{d_i})$, $\tilde \Phi = \bigoplus_{i=1}^n (\tilde \phi_i I_{d_i})$, and {$\tilde \Psi_2 = \bigoplus_{i=1}^n \bigoplus_{j\in\mathcal N_i(0)} \tilde \psi_{ij}$}. Define the nonnegative matrix
\begin{equation*}
\tilde M
=
\begin{bmatrix}
\tilde D_1 + \bar \psi I + r I	& B_1
\\
\Psi_1							& B_2 + \tilde D_2 + \tilde \Phi + \tilde \Psi_2
\end{bmatrix}. 
\end{equation*}
Then, for a given $\alpha > 0$, the following statements are equivalent:
\begin{itemize}
\item There exist $\phi_1, \dotsc, \phi_n \in [\ubar{\phi}, \bar{\phi}]$ such that $\eta(M) \leq  -\alpha$.

\item There exist $\tilde \phi_1, \dotsc, \tilde \phi_n \in [r-\bar{\phi}, r-\ubar{\phi}]$ such that $\eta(\tilde M) + \alpha \leq {\psi_0} + \bar \delta + r$.
\end{itemize}
Moreover, between $\{\phi_i\}_{i=1}^n$ and $\{\tilde \phi_i\}_{i=1}^n$, there is a one-to-one correspondence given by the equation
\begin{equation}\label{eq:phi<->til_phi}
\phi_i = r- \tilde  \phi_i. 
\end{equation}
\afterequation
\end{proposition}

\begin{proof}
Assume that there exist $\phi_1, \dotsc, \phi_n \in [\ubar{\phi}, \bar{\phi}]$ satisfying $\eta(M) \leq  -\alpha$. Define $\tilde \phi_i$ by \eqref{eq:phi<->til_phi}. Then we see that $\tilde M = M + (\bar{\delta} + r + \bar \psi)I$. This implies $\eta(\tilde M) + \alpha \leq \bar{\delta} + r + {\bar \psi}$. We also have $\tilde \phi_i\in [r-\bar{\phi}, r-\ubar{\phi}]$. The other direction can be shown in the same way and, hence, its proof is omitted.
\end{proof}

Using Proposition~\ref{prop:tildeA}, we can reduce Problem~\ref{prb:} to a geometric program under Assumption~\ref{assm:hetero}, as stated in the following theorem:

\begin{theorem}
Let $\tilde \phi_1$, $\dotsc$, $\tilde \phi_n$ be the solutions to the following geometric program:
\begin{subequations}\label{eq:GP}
\begin{align}
&\minimize_{\tilde{\phi_i},\,v}\ 
\sum_{i=1}^n F(\tilde \phi_i)\label{dummy:i3c}
\\
&\subjectto\ 
 \left(\tilde M + \alpha I\right) v \leq 
 \left(\bar \delta +r + \bar \psi\right)v,\hspace{.5cm} \label{eq:PF1}
\\
& \phantom{\subjectto}\ v > 0,  \label{eq:PF2}
\\
& \phantom{\subjectto}\ r - \bar \phi \leq \tilde \phi_i \leq r - \ubar \phi. \label{dummy:k1-}
\end{align}
\end{subequations}
Then, $\{\phi_i\}_{i=1}^n$, defined in \eqref{eq:phi<->til_phi} solve Problem~\ref{prb:}.
\end{theorem}

\begin{proof}
By Proposition~\ref{prop:tildeA}, Problem~\ref{prb:} is equivalent to the optimization problem
\begin{align}
\minimize_{\tilde \phi_i}\ 
&\sum_{i=1}^n f(r- \tilde{\phi_i}) \notag
\\
\subjectto\ 
& \eta(\tilde M+\alpha I)  \leq  \bar \delta + r + {\bar \psi}, \label{eq:eta(tildeM...}
\\
&r - \bar \phi \leq \tilde \phi_i \leq r - \ubar \phi \notag,
\end{align}
after the change of variables~\eqref{eq:phi<->til_phi}. Minimizing the objective
function in this problem is equivalent to minimizing the one in \eqref{eq:GP} by
the definition of $F$, whose constant term $s$ can be ignored in the
optimization. Then, since $\tilde M + \alpha I$ is irreducible by
Proposition~\ref{prop:irreducible}, we can replace the constraint
\eqref{eq:eta(tildeM...} into \eqref{eq:PF1} and \eqref{eq:PF2} in the same way
as in~\cite{Preciado2014} using Perron-Frobenius lemma. Also, by a similar
argument as in~\cite{Preciado2014}, we can show that \eqref{eq:GP} is a
geometric program. This is because $F$ is a posynomial and each entry of the
matrix~{$\tilde M + \alpha I$} is a posynomial in the variables
$\tilde{\phi}_1$, $\dotsc$, $\tilde{\phi}_n$. The details are omitted.
\end{proof}

\begin{remark}
When $\psi_{ij}$ are also design variables, the above argument reduces Problem~\ref{prb:} to a signomial program, which are (in general) hard to solve~\cite{Boyd2007}.
\end{remark}

\section{Numerical Results} 

We illustrate our results with a numerical example. Let $\mathcal G_0$ be the
graph of a social network of $n=247$ nodes and {$m = 940$} edges.  The
adjacency matrix of the graph has spectral radius $\rho = 13.53$. We assume that
all nodes have identical recovery rate $\delta = 0.1$ and infection rate
{$\beta = \delta/(1.1 \rho) =6.720 \times 10^{-3}$}. Since $\delta/\beta =
(1.1) \rho > \rho$, Theorem~\ref{thm:stbl:homo} does not guarantee the stability
of the infection-free equilibrium when $\phi = 0$, i.e., when the network does
not adapt.

Let us design the cost-optimal cutting rates so that the spread stabilizes towards the disease-free equilibrium in the adaptive network. We assume $\ubar \phi = 0$, $\bar \phi = 4\beta$, and $\psi_{ij} = \beta$ and use the following cost function in our numerical simulations:
\begin{equation*} 
f(x) 
= 
\frac{(r-x)\mathstrut^{-1} - (r-\ubar{\phi})^{-1}}{(r-\bar{\phi})^{-1} - (r-\ubar{\phi})^{-1}}. 
\end{equation*}
We have chosen this function since it is increasing and presents diminishing returns. Also, we have normalized it, so that $f(\ubar{\phi}) = 0$ and $f(\bar{\phi}) = 1$, and fixed $r = 2\bar{\phi}$. Let the desired exponential decay rate be $\alpha = 0.005$ and solve the geometric program in Theorem~\ref{thm:optim:homo} to obtain the optimal cutting rates~$\phi_i$. Fig.~\ref{fig:ExPhi} shows a scatter plot for the optimal rates, $\phi_i$, versus the degrees of the nodes for all $i\in\mathcal{V}$. 
\begin{figure}[tb]
\centering
\includegraphics[width=.475\textwidth]{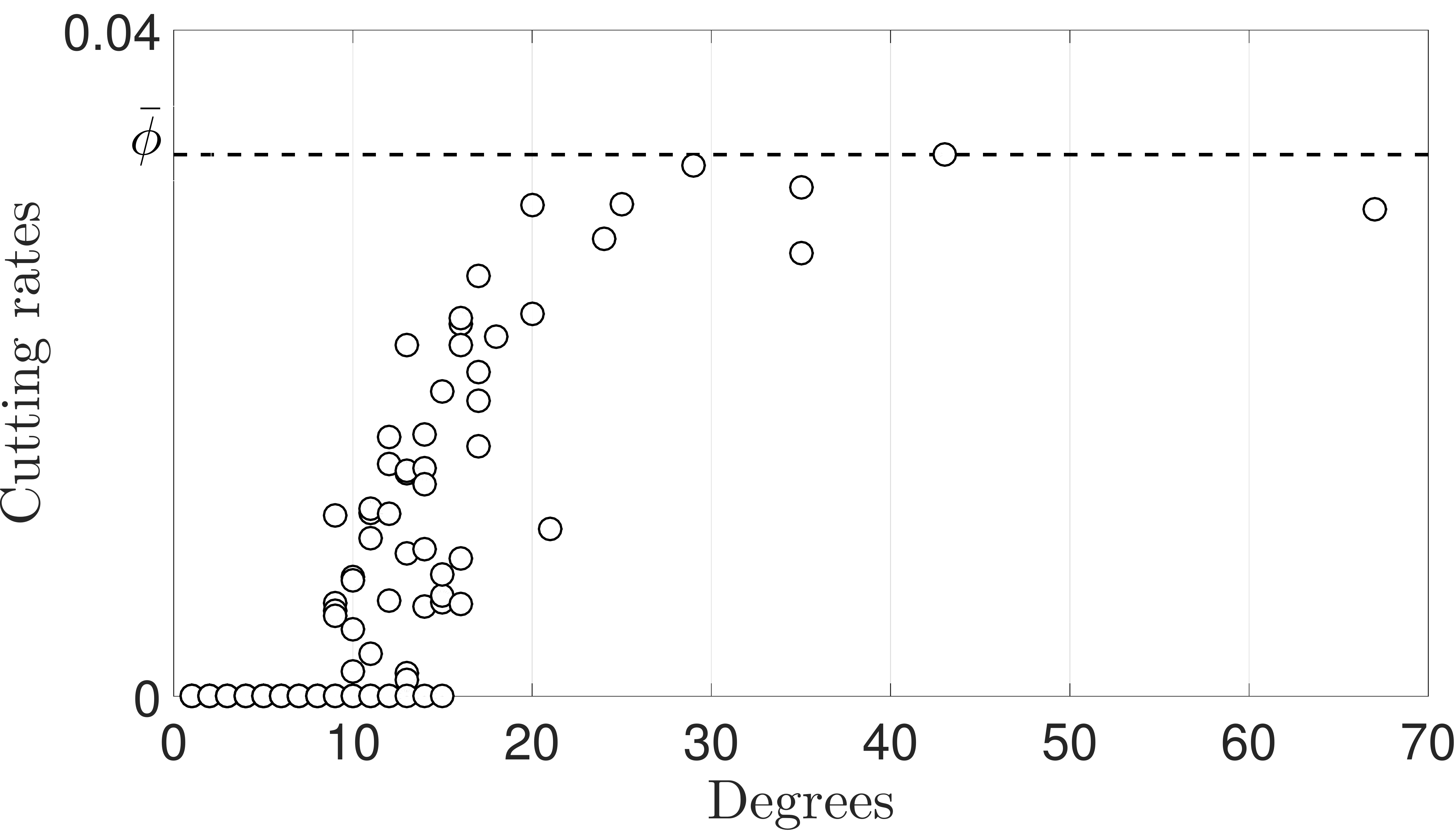}\hspace*{.5cm}
\caption{Cost-optimal cutting rates for stabilization}
\label{fig:ExPhi}
\end{figure}
The resulting switching policy suggests that, in general, nodes with a larger degree should have higher cutting rates (as could be naturally expected). However, the relationship between the optimal cutting rates and the degrees is not trivial. Alternatively, we have also studied the relationship between cutting rates and other network centrality measures and $K$-scores (though we omit these figures for space limitations). Our simulations do not show any trivial dependency between cutting rates and any of the measures considered.

\section{Conclusion}

In this paper, we have studied the dynamics of spreading processes taking place
in networks that adapt their structure depending on the state of the dynamics.
Our model is based on a collection of stochastic differential equations with
Poisson jumps that model the joint evolution of the states of the process taking
place in the network, as well as the evolution of the network structure. To
illustrate our framework, we have focused our attention in a popular model of
spreading dynamics, the SIS model, and study it dynamics over adaptive, switched
networks. For this particular model, we have derived conditions for the dynamics
of the spread to converge towards the disease-free equilibrium. Using this
stability result, we have then formulated an optimization program to find the
cost-optimal adaptive strategy to achieve stability. We have also showed that
this optimization program can be efficiently solved  using geometric
programming. A numerical example was included to illustrate our results. An
interesting future work is to fully investigate the difference of information
structures in our model and the one in \cite{Guo2013} addressed in
Remark~\ref{rmk:}.


\end{document}